\documentclass[a4paper,UKenglish]{lipics-v2018}

\usepackage{microtype}%if unwanted, comment out or use option "draft"

\usepackage{todonotes,comment}
\usepackage[utf8]{inputenc}

\usepackage{xspace}
\usepackage{mathtools,amssymb}

\usepackage{floatflt}

\usepackage{booktabs}
\usepackage{stmaryrd}
\usepackage{xspace}
\usepackage{esvect}

\usepackage{algorithm}
\usepackage[noend]{algpseudocode}

\usepackage{listings}
  
  \lstdefinelanguage{pseudo}{
    morekeywords={if,elseif,then,return,end,choose,guess,when,for,foreach,case},
    morekeywords=[3]{false,true,and,or,not},
    morecomment=[l]{//}
  }
\newcommand{\FV}{\mathrm{FV}}
\newcommand{\J}{\mathcal{J}}
\newcommand{\B}{\mathcal{B}}

\newcommand{\dfn}{\mathrel{\mathop:}=}
\newcommand{\ddfn}{\mathrel{\mathop{{\mathop:}{\mathop:}}}=}
\newcommand{\ar}[1]{\mathrm{ar}(#1)}
\newcommand{\rel}[1]{\mathrm{Rel}(#1)}
\newcommand{\dif}{\mathrm{diff}}
\newcommand{\simi}{\mathrm{sim}}
\newcommand{\SIM}{\mathrm{SIM}}

\newcommand{\N}{\mathbb N}
\newcommand{\primes}{\mathbb P}
\newcommand{\mA}{\mathfrak A}
\newcommand{\mB}{\mathfrak B}

\newcommand{\size}{\mathrm{size}}
\newcommand{\Ha}{\mathrm H}
\newcommand{\la}{\mathrm L}
\newcommand{\pa}{\mathrm P}

\newcommand{\NP}{\protect\ensuremath{\mathrm{NP}}}
\newcommand{\NLIN}{\protect\ensuremath{\mathrm{NLIN}}}

\newcommand{\PSPACE}{\protect\ensuremath{\mathrm{PSPACE}}} 
\newcommand{\NL}{\protect\ensuremath{\mathrm{NLOGSPACE}}}

\newcommand{\NSPACE}{\protect\ensuremath{\mathrm{NSPACE}}}

\newcommand{\FO}{\protect\ensuremath{\mathrm{FO}}}
\newcommand{\MSO}{\protect\ensuremath{\mathrm{MSO}}}
\newcommand{\CMSO}{\protect\ensuremath{\mathrm{CMSO}}}
\newcommand{\EFO}{\protect\ensuremath{\exists\mathrm{FO}}}
\newcommand{\fo}{\FO}
\newcommand{\SO}{\protect\ensuremath{\mathrm{SO}}}
\newcommand{\STC}{\protect\ensuremath{\mathrm{2TC}}}
\newcommand{\FTC}{\protect\ensuremath{\mathrm{1TC}}}
\newcommand{\TC}{\protect\ensuremath{\mathrm{TC}}}
\newcommand{\SOTC}{\protect\ensuremath{\SO(\TC)}}
\newcommand{\MSOTC}{\protect\ensuremath{\MSO(\TC)}}
\newcommand{\CMSOTC}{\protect\ensuremath{\CMSO(\TC)}}
\newcommand{\EPSOTC}{\protect\ensuremath{\SO(\STC)[\exists]}}

\theoremstyle{plain}
\newtheorem{proposition}[theorem]{Proposition}
\usepackage{booktabs}   %% For formal tables:
\usepackage{subcaption} %% For complex figures with subfigures/subcaptions
\title{Expressivity within second-order transitive-closure logic} 
\author{Flavio Ferrarotti}{Software Competence Center Hagenberg,
Hagenberg,
Austria}{flavio.ferrarotti@scch.at}{https://orcid.org/0000-0003-2278-8233}{}
\author{Jan Van den Bussche}{Hasselt University, Hasselt,
Belgium}{jan.vandenbussche@uhasselt.be}{https://orcid.org/0000-0003-0072-3252}{}
\author{Jonni Virtema}{Hasselt University, Hasselt, Belgium}{jonni.virtema@uhasselt.be}{https://orcid.org/0000-0002-1582-3718}{}

\authorrunning{F. Ferrarotti, J. Van den Bussche, J. Virtema} 
\Copyright{Flavio Ferrarotti, Jan Van den Bussche and Jonni Virtema}

\subjclass{Theory of computation $\rightarrow$ Finite Model Theory}% mandatory: Please choose ACM 1998 classifications from http://www.acm.org/about/class/ccs98-html . E.g., cite as "F.1.1 Models of Computation". 
% mandatory: Please provide 1-5 keywords

\funding{The
research reported in this paper results from the joint project {\em
Higher-Order Logics and Structures} supported by the Austrian
Science Fund (FWF: \textbf{[I2420-N31]}) and the Research
Foundation Flanders (FWO: \textbf{[G0G6516N]}).}%optional, to capture a funding statement, which applies to all authors. Please enter author specific funding statements as fifth argument of the \author macro.

\EventEditors{John Q. Open and Joan R. Access}
\EventNoEds{2}
\EventLongTitle{42nd Conference on Very Important Topics (CVIT 2016)}
\EventShortTitle{CVIT 2016}
\EventAcronym{CVIT}
\EventYear{2016}
\EventDate{December 24--27, 2016}
\EventLocation{Little Whinging, United Kingdom}
\EventLogo{}
\SeriesVolume{42}
\ArticleNo{1}
\nolinenumbers %uncomment to disable line numbering
\hideLIPIcs  %uncomment to remove references to LIPIcs series (logo, DOI, ...), e.g. when preparing a pre-final version to be uploaded to arXiv or another public repository
\begin{document}

\maketitle

\begin{abstract}
Second-order transitive-closure logic, $\SOTC$, is an expressive declarative language that captures the complexity class $\PSPACE$. Already its monadic fragment, $\MSOTC$, allows the expression of various $\NP$-hard and even $\PSPACE$-hard problems in a natural and elegant manner. As $\SOTC$ offers an attractive framework for expressing properties in terms of declaratively specified computations, it is interesting to understand the expressivity of different features of the language. This paper focuses on the fragment $\MSOTC$, as well on the purely existential fragment $\SO(\STC)(\exists)$; in $\STC$, the $\TC$ operator binds only tuples of relation variables. We establish that, with respect to expressive power, $\SO(\STC)(\exists)$ collapses to existential first-order logic. In addition we study the relationship of $\MSOTC$ to an extension of $\MSOTC$ with counting features ($\CMSO(\TC)$) as well as to order-invariant $\MSO$. We show that the expressive powers of $\CMSO(\TC)$ and $\MSOTC$ coincide. Moreover we establish that, over unary vocabularies, $\MSOTC$ strictly subsumes order-invariant $\MSO$.
\end{abstract}

\section{Introduction}

Second-order transitive-closure logic, $\SOTC$, is an expressive declarative language that captures
the complexity class $\PSPACE$~\cite{HarelP84}. It extends second-order logic with a transitive closure operator over relations of relations, i.e., over super relations among relational structures. The super relations are defined by means of second-order logic formulae with free relation variables.  Already its monadic fragment, $\MSOTC$, allows the expression of $\NP$-complete problems in a natural and elegant manner. Consider, for instance, the well known Hamiltonian cycle query over the standard vocabulary of graphs, which is not expressible in monadic second-order logic~\cite{Courcelle95}. 
\begin{example}
%In MSO(TC) we can simply say that
A graph $G = (V, E)$ has a Hamiltonian cycle if the following holds:
\begin{enumerate}[a.] 
%\item There is a super relation $\cal R$ such that $(V, v, V', v') \in {\cal R}$ iff $V' \setminus V = \{v'\}$ and $(v, v')$ is an edge in $E$.
\item There is a relation $\cal R$ such that $(Z, z, Z', z') \in {\cal R}$ iff $Z' = Z \cup \{z'\}$, $z'\notin Z$, and $(z, z')\in E$.
\item The tuple $(\{x\}, x, V, y)$ is in the transitive closure of $\cal R$, for some $x,y\in V$ s.t. $(y, x) \in E$.
%\item The tuple $(X, x, Y, y)$, where $X = \{x\}$, $Y$ is the vertex set $V$ and $(x, y) \in E$, is in the transitive closure of $\cal R$.
\end{enumerate}      
In the language of $\MSOTC$ this can be written as follows:
\[\exists X Y x y \big(X(x) \wedge \forall z (z \neq x \rightarrow \neg X(x)) \wedge \forall z (Y(z)) \wedge E(y, x) \wedge [\TC_{Z, z, Z', z'} \varphi](X, x, Y, y)\big),\] 
where $\varphi \dfn \neg Z(z') \wedge \forall x \big(Z'(x) \leftrightarrow (Z(x) \vee z' = x)\big) \wedge E(z, z')$.
\end{example}
Even some well-known $\PSPACE$-complete problems such as deciding whether a given quantified Boolean formula QBF is valid~\cite{lad77} can be expressed in MSO(TC) (see Section~\ref{sec:complexityMSOTC}). 

In general, SO(TC) offers an attractive framework for expressing properties in terms of declaratively specified computations at a high level of abstraction. There are many examples of graph computation problems that involve complex conditions such as graph colouring \cite{abukhzam:cocoon2003}, topological subgraph discovery \cite{grohe:stoc2011}, recognition of hypercube graphs \cite{FerrarottiRT14}, and many others (see \cite{bollobas:2002,ferrarotti:2008,FerrarottiGT17}). Such graph algorithms are difficult to specify, even by means of rigorous methods such as Abstract State Machines (ASMs) \cite{boerger:2003}, B \cite{abrial:2005} or Event-B \cite{abrial:2010}, because the algorithms require the definition of characterising conditions for particular subgraphs that lead to expressions beyond first-order logic. Therefore, for the sake of easily comprehensible and at the same time fully formal high-level specifications, it is reasonable to explore languages such as SO(TC). Let us see an example that further supports these observations.  

\begin{example}\label{ex:self_similar}
Self-similarity of complex networks~\cite{song:nat2005} (aka scale invariance) has practical applications in diverse areas such as the world-wide web~\cite{Dill:2002}, social networks~\cite{Guimera2003}, and biological networks~\cite{reka:jcs2005}. Given a network represented as a finite graph $G$, it is relevant to determine whether $G$ can be built starting from some graph pattern $G_p$ by recursively replacing nodes in the pattern by new, ``smaller scale'', copies of $G_b$. If this holds, then we say that $G$ is self-similar. 
%Examples of self-similar graphs are shown in Figure \ref{figure:self-similarity}. Note that they are constructed using a triangle as graph pattern. 

Formally, a graph $G$ is \emph{self-similar} w.r.t. a graph pattern $G_p$ of size $k$, if there is a sequence of graphs $G_0, G_1, \ldots, G_n$ such that $G_0 = G_p$, $G_n = G$ and, for every pair $(G_{i}, G_{i+1})$ of consecutive graphs in the sequence, there is a partition $\{P_1, \ldots, P_k\}$ of the set of nodes of $G_{i+1}$ which satisfies the following: 
\begin{enumerate}[a.]
\item For every $j = 1, \ldots, k$, the sub-graph induced by $P_j$ in $G_{i+1}$ is isomorphic to $G_{i}$. 
\item There is a graph $G_t$ isomorphic to $G_p$ with set of nodes $V_t = \{a_1, \ldots, a_k\}$ for some $a_1 \in P_1, \ldots, a_k \in P_k$ and set of edges \[E_t =\{(a_i,a_j) \mid \text{there is an edge}\; (x, y) \; \text{of} \; G_{i+1} \; \text{such that} \; P_i(x) \; \text{and} \; P_j(y)\}.\]
\item For very $1 \leq i < j \leq k$, the closed neighborhoods $N_{G_{i+1}}[P_i]$ and $N_{G_{i+1}}[P_j]$ of $P_i$ and $P_j$ in $G_{i+1}$, respectively, are isomorphic.  
\end{enumerate}
It is straightforward to write this definition of self-similarity in SO(TC), for we can clearly write a second-order logic formula which defines such a super relation $\cal R$ on graphs and then simply check whether the pair of graphs $(G, G_p)$ is in the transitive closure of $\cal R$.
\end{example}

Highly expressive query languages are gaining relevance in areas such as knowledge representation (KR), rigorous methods and provers. There are several examples of highly  expressive  query  languages  related to applications in KR. See for instance the monadically defined queries in~\cite{RudolphK13}, the Monadic Disjunctive SNP queries in~\cite{BienvenuCLW13} or the guarded queries in~\cite{BourhisKR15}.  All of them can be considered fragments of Datalog. Regarding rigorous methods, the TLA$^+$ language~\cite{lamport:2002} is able to deal with higher-order formulations, and tools such as the TLA$^+$ Proof System\footnote{\url{https://tla.msr-inria.inria.fr/tlaps}} and the TLA$^+$ Model-Checker (TLC)\footnote{\url{https://lamport.azurewebsites.net/tla/tlc.html}} can handle them (provided a finite universe of values for TLC). Provers such as Coq\footnote{\url{https://coq.inria.fr/}} and Isabelle\footnote{\url{https://isabelle.in.tum.de/}} can already handle some high-order expression. 
Moreover, the success with solvers for the Boolean satisfiability problem (SAT) has encouraged researchers to target larger classes of problems, including $\PSPACE$-complete problems, such as satisfiability of Quantified Boolean formulas (QBF). Note the competitive evaluations of QBF solvers (QBFEVAL) held in 2016 and~2017 and recent publications on QBF solvers such as~\cite{BlinkhornB17,PeitlSS17,HeuleSB17} among several others. 

We thus think it is timely to study which features of highly expressive query languages affect their expressive power. In this sense, $\SOTC$ provides a good theoretical base since, apart from been a highly expressive query language (recall that it captures $\PSPACE$), it enables natural and precise high-level definitions of complex practical problems, mainly due to its ability to express properties in terms of declaratively specified computations.  While second-order logic extended with the standard partial fixed-point operator, as well as first-order logic closed under taking partial fixed-points and under an operator for non-deterministic choice, also capture the class of $\PSPACE$ queries over arbitrary finite structure~\cite{Richerby04}, relevant computation problems such as that in Example~\ref{ex:self_similar} are clearly more difficult to specify in these logics. The same applies to the extension of first-order logic with the partial fixed-point operator, which is furthermore subsumed by $\SOTC$ since it captures $\PSPACE$ only on the class of ordered finite structures~\cite{AbiteboulV89}. Note that $\SOTC$ coupled with hereditary finite sets and set terms, could be considered as a kind of declarative version of Blass, Gurevich, and Shelah (BGS) model of abstract state machine~\cite{BLASS1999141}, which is a powerful language in which all computable queries to relational databases can be expressed~\cite{BLASS200220}.   

Our results can be summarized as follows.
\begin{enumerate}
\item
We investigate to what extent universal quantification and
negation are important to the expressive power of $\SOTC$.  Specifically, we
consider the case where TC-operators are applied only to
second-order variables.  Of course, a second-order variable can
simulate a first-order variable, since we can express already in
first-order logic (FO) that a set is a singleton.  This, however,
requires universal quantification.

We define a ``purely existential'' fragment of $\SOTC$,
$\SO(\STC)(\exists)$, as the fragment without universal
quantifiers and in which $\TC$-operators occur only positively
and bind only tuples of relation variables.  We show that the expressive
power of this fragment collapses to that of existential FO.

For SO alone, this collapse is rather obvious and was already
remarked by Rosen in the introduction of
his paper \cite{rosen_existential}.  Our result generalizes
this collapse to include TC operators, where it is no longer obvious.

\item
We investigate the expressive power of the monadic fragment,
$\MSOTC$.  On strings, this logic is equivalent to the complexity
class NLIN\@.  Already on unordered structures, however, we show 
that $\MSOTC$ can express counting terms and numeric predicates 
in $\NL$.
In particular,
$\MSOTC$ can express queries not expressible in the fixpoint logic FO(LFP)\@.
We also discuss the fascinating
open question whether the converse holds as
well.

\item
We compare the expressive power of $\MSOTC$ to
that of order-invariant $\MSO$.  Specifically, we show that
$\MSOTC$ can express queries not expressible in order-invariant
$\MSO$; over monadic vocabularies, we show that order-invariant
MSO is subsumed by $\MSOTC$.  Again, what happens over
higher-arity relations is an interesting open question.
\end{enumerate}

This paper is organized as follows.
In Section \ref{sec:preli} definitions and basic notions related
to $\SOTC$ are given.
In Section \ref{sec:complexityMSOTC}
the complexity of model checking is studied.
Section \ref{sec:epo} is dedicated to
establishing the collapse of $\SO(\STC)(\exists)$ to existential
first-order logic. Sections \ref{sec:counting} and
\ref{sec:order} concentrate on the relationships between $\MSOTC$
and the counting extension $\CMSO(\TC)$ and order-invariant
$\MSO$, respectively.  We conclude with a discussion of open
questions in Section~\ref{seconc}.

\section{Preliminaries}\label{sec:preli}
We assume that the reader is familiar with finite model theory, see e.g., \cite{FMTEF} for a good reference. 
For a tuple $\vec{a}$ of elements, we denote by $\vec{a}[i]$ the $i$th element of the tuple. 
%We consider only relational and purely relational vocabularies and finite structures.
%
%A relational vocabulary $\tau$ consist of first-order variable symbols $x_i$ and  second-order variable symbols $R_i$ with prescribed arities $\ar{R_i}\in\N$.
%A purely relational vocabulary is a relational vocabulary without first-order variable symbols.
%
%
%
We recall from the literature, the syntax and semantics of first-order (\FO) and
second-order (\SO) logic, as well as their extensions with the
transitive closure operator (\TC). We assume a sufficient supply of \emph{first-order} and \emph{second-order} \emph{variables}. The natural number $\ar{R}\in\N$, is the \emph{arity} of the second-order variable $X$. By \emph{variable}, we mean either a first-order or second-order variable.
Variables $\chi$ and $\chi'$ have the same \emph{sort} if either both $\chi$ and $\chi'$ are first-order variables, or both are second-order variables of the same arity.  Tuples $\vec{\chi}$ and $\vec{\chi'}$ of variables have the same \emph{sort}, if the lengths of $\vec{\chi}$ and $\vec{\chi'}$ are the same and, for each $i$, the sort of $\vec{\chi}[i]$ is the same as the sort of $\vec{\chi}[i]$.

\begin{definition}
The formulas of $\SOTC$ are defined by the following grammar:
\[
\varphi \ddfn x=y \mid X(x_1,\dots,x_k) \mid \neg \varphi \mid (\varphi \lor \varphi) \mid \exists x \varphi \mid \exists Y \varphi \mid [\TC_{\vec{X},\vec{X'}}\varphi](\vec{Y},\vec{Y'}),
\]
where $X$ and $Y$ are second-order variables, $k=\ar{X}$, $x, y,x_1, \dots, x_k$ are first-order variables, $\vec{X}$ and $\vec{X'}$ are disjoint tuples of variables of
%the same length and
the same sort, and $\vec{Y}$ and $\vec{Y'}$
are also tuples of variables of that same sort (but not necessarily
disjoint).
\end{definition}
The set of free variables of a formula $\varphi$, denoted by $\FV(\varphi)$ is defined as usual. For the $\TC$ operator, we define
\[
    \FV([\TC_{\vec{X},\vec{X'}}\varphi](\vec{Y},\vec{Y'})) \dfn
    (\FV(\varphi) - (\vec{X} \cup \vec{X'})) \cup \vec{Y} \cup
    \vec{Y'}.
\]    
Above in the right side, in order to avoid cumbersome notation, we use $\vec{X}$, $\vec{X'}$, $\vec{Y}$ and $\vec{Y'}$ to denote the sets of variables occurring in the tuples.

A \emph{vocabulary} is a finite set of variables.
A (finite) \emph{structure} $\mA$ over a vocabulary $\tau$
is a pair $(A,I)$, where
$A$ is a finite nonempty set called the \emph{domain} of $\mA$, and
$I$ is an \emph{interpretation} of $\tau$ on $A$.  By this we
mean that whenever $x\in\tau$
is a first-order variable, then $I(x) \in A$, and whenever $X\in\tau$ is a
second-order variable of arity $m$, then $I(X) \subseteq A^m$.
In this article, structures are always finite.
We denote $I(X)$ also by $X^{\mA}$. For a variable $X$ and a suitable value $R$ for that variable, $\mA[R/X]$ denotes the structure over $\tau\cup\{X\}$ equal to $\mA$ except that $X$ is mapped to $R$. We extend the notation also to tuples of variables and values, $\mA[\vec{X}/\vec{R}]$, in the obvious manner.
We say that a
vocabulary $\tau$
is \emph{appropriate} for a
formula $\varphi$
if $\FV(\varphi) \subseteq \tau$. 

\begin{definition}
Let $\mA$ be a structure
over $\tau$ and $\varphi$ an  $\SOTC$-formula such that $\tau$ is appropriate for $\varphi$.
The satisfaction of $\varphi$ by $\mA$, denoted by $\mA
\models \varphi$, is defined as follows.  We only give the cases for
second-order quantifiers and transitive closure operator; the remaining cases are defined as usual.
\begin{itemize}
\item For second-order variable $X$: $\mA\models\exists X  \varphi$  iff  $\mA[R/X] \models \varphi$, for some $R\subseteq A^{\ar{X}}$.
\item For the case of the $\TC$-operator, consider a formula $\psi$ of the form $[\TC_{\vec{X},\vec{X'}}\varphi](\vec{Y},\vec{Y'})$ and  let
    $\mA=(A,I)$.  Define $\J_{\vec{X}}$ to be the following set
\[
     \{J(\vec{X}) \mid \text{$J$ is an interpretation of $\vec{X}$ on $A$}\}  = \{J(\vec{X'}) \mid \text{$J$ is an interpretation of $\vec{X'}$ on $A$}\}
     \]
and consider the binary relation $\B$ on $\J_{\vec{X}}$ defined as
    follows: 
    \[
    \B \dfn \{(\vec{R},\vec{R'}) \in \J_{\vec{X}}\times \J_{\vec{X}} \mid  \mA[\vec{R}/\vec{X},\vec{R'}/\vec{X'}] \models \varphi\}.
    \]
    We set $\mA \models \psi$ to hold if
    $(I(\vec{Y}),I(\vec{Y'}))$ belongs to the transitive closure
    of $\B$.
Recall that, for a binary relation $\B$ on any set $\J$, the
transitive closure of $\B$ is defined by
\begin{align*}
\TC(\B) \dfn &\{ ({a}, {b}) \in \J\times \J
\mid \text{$\exists n > 0$ and ${e}_0,\dots,{e}_n\in \J$}\\
& \text{such that ${a}={e}_0$, ${b}={e}_n$, and
  $({e}_i,{e}_{i+1})\in \B$ for all $i<n$}
 \}.
\end{align*}
\end{itemize}
\end{definition}
By $\TC^m$ we denote the variant of $\TC$  in which the quantification of $n$ above is restricted to natural numbers $\leq m$. That is,  $\TC^m(\B)$ consists of pairs $(\vec{a}, \vec{b})$ such that $\vec{b}$ is reachable from $\vec{a}$ by $\B$ in at most $m$ steps. Moreover, by $\STC$ and $\STC^m$ we denote the syntactic restrictions of $\TC$ and $\TC^m$ of the form
\[
[\TC_{\vec{X},\vec{X'}}\varphi](\vec{Y},\vec{Y'}) \text{ and } [\TC^m_{\vec{X},\vec{X'}}\varphi](\vec{Y},\vec{Y'}),
\]
where $\vec{X}$, $\vec{X'}$, $\vec{Y}$, $\vec{Y'}$ are tuples of second-order variables (i.e. without first-order variables). The logic $\SO(\STC)$ then denotes the extension of second-order logic with $\STC$-operator. Analogously, by $\FO(\FTC)$, we denote the extension of first-order logic with applications of such transitive-closure operators that bind only first-order variables.\footnote{In the literature $\FO(\FTC)$ is often denoted by $\FO(\TC)$.}

\section{Complexity of MSO(TC)}\label{sec:complexityMSOTC}

The descriptive complexity of different logics with the
transitive closure operator has been thoroughly studied by Immerman. Let $\SO(\mathrm{arity}\, k)(\TC)$ denote the fragment of $\SOTC$ in which second-order variables are all of arity $\leq k$.
\begin{theorem}[\cite{Immerman:1987, Immerman88}] \label{thm:fotc}~
\begin{itemize}
\item On finite ordered structures, first-order transitive-closure logic $\FO(\FTC)$ captures nondeterministic logarithmic space $\NL$.
\item On strings (word structures),
$\SO(\mathrm{arity}\, k)(\TC)$ captures the complexity class $\NSPACE(n^k)$.
\end{itemize}
\end{theorem}
See also the discussion in the conclusion section.

By the above theorem, $\MSOTC$ captures nondeterministic linear space
$\NLIN$ over strings.
Deciding whether a given \emph{quantified Boolean formula} is
valid (QBF) is a well-known $\PSPACE$-complete problem
\cite{lad77}. Observe that there are $\PSPACE$-complete problems
already in $\NLIN$; in fact QBF is such a problem. Thus,
we can conclude the following.
The inclusion in $\PSPACE$ is clear.

\begin{proposition}
Data complexity of $\MSOTC$ is $\PSPACE$-complete.
\end{proposition}

We next turn to combined complexity of model checking.
By the above proposition, this is at least 
$\PSPACE$-hard.  However, the straightforward algorithm for model
checking $\MSOTC$ clearly has polynomial-space combined
complexity.  We thus conclude:
\begin{proposition}
Combined complexity of $\MSOTC$ is $\PSPACE$-complete.
\end{proposition}

For combined complexity, we can actually sharpen
the $\PSPACE$-hardness;
already a very simple fragment of $\MSOTC$ is $\PSPACE$-complete.

Specifically,
we give a reduction from the \emph{corridor tiling} problem, which is a well-known $\PSPACE$-complete problem.
Instance of the corridor tiling problem is a tuple $P=(T,H,V,\vec{b},\vec{t},n)$, where $n\in \N$ is a positive natural number, $T=\{1,\dots,k\}$, for some $k\in\N$, is a finite set of \emph{tiles}, $H,V\subseteq T\times T$ are \emph{horizontal} and \emph{vertical constraints}, and $\vec{b},\vec{t}$ are $n$-tuples of tiles from $T$. A \emph{corridor tiling for $P$} is a function $f: \{1,\dots,n\} \times \{1\dots,m\} \rightarrow T$, for some $m\in \N$, such that 
\begin{itemize}
\item $\big( f(1,1), \dots f(n,1) \big) = \vec{b}$ and  $\big( f(1,m), \dots f(n,m) \big) = \vec{t}$,
\item $\big( f(i,j),f(i+i,j) \big) \in H$, for $i<n$ and $j\leq m$, 
\item $\big( f(i,j),f(i,j+1) \big) \in V$, for $i\leq n$ and $j< m$.
\end{itemize}

\noindent The \emph{corridor tiling problem} is the following $\PSPACE$-complete decision problem \cite{Chlebus:1986}:\\
\textbf{Input:} An instance $P=(T,H,V,\vec{b},\vec{t},n)$ of the corridor tiling problem.\\
\textbf{Output:} Does there exist a corridor tiling for $P$?

Let \emph{monadic $\STC[\forall\FO]$} denote the fragment of $\MSO(\STC)$ of the form $[\TC_{\vec{X},\vec{X'}}\varphi] (\vec{Y},\vec{Y'})$, where $\psi$ is a formula of universal first-order logic. 

\begin{theorem}
Combined complexity of model checking for monadic $\STC[\forall\FO]$ is $\PSPACE$-complete.
\end{theorem}
\begin{proof}
Inclusion to $\PSPACE$ follows from the corresponding result for $\MSOTC$. In order to prove hardness,
we give a reduction from corridor tiling. Let $P=(T,H,V,\vec{b},\vec{t},n)$ be an instance of the corridor tiling problem and set $k\dfn \lvert T\rvert$. Let $\tau= \{s,X_1,\dots X_k, Y_1,\dots,Y_k\}$ be a vocabulary, where $s$ is a binary second-order variable and $X_1,\dots X_k, Y_1,\dots,Y_k$ are monadic second-order variables. Let $\mA_P$ denote the structure over $\tau$ such that $A=\{1,\dots,n\}$, $I(s)$ is the canonical successor relation on $A$, and, for each $i\leq k$,  $I(X_i)=\{j\in A \mid \vec{b}[j]=i\}$ and $I(Y_i)=\{j\in A \mid \vec{t}[j]=i\}$. Define
\begin{align*}
\varphi_H &\dfn \forall x y \big (s(x,y) \rightarrow \bigvee_{(i,j)\in H} Z'_i(x)\land Z'_j(y) \big), \quad \varphi_V \dfn \forall x \bigvee_{(i,j)\in V} Z_i(x)\land Z'_j(y) \\
\varphi_T &\dfn \forall x \bigvee_{i\in T} \big( Z'_i(x) \land \bigwedge_{j\in T, i\neq j} \neg Z'_j(x) \big),
\end{align*}
where $\vec{Z}$ and $\vec{Z'}$ $k$-tuples of distinct monadic second-order variables not in $\tau$. 
We then define
\(
\varphi_P \dfn \TC_{\vec{Z},\vec{Z'}}[\varphi_T \land \varphi_H\land \varphi_V] (\vec{X}, \vec{Y}).
\)
We claim that $\mA_P\models \varphi_P$ if and only if there exists a corridor tiling for $P$, from which the claim follows.
\end{proof}

\section{Existential positive SO(2TC) collapses to EFO}\label{sec:epo}
Let $\EPSOTC$ denote the syntactic fragment of $\SO(\STC)$ in which existential quantifiers and the $\TC$-operator occur only positively, that is, in scope of even number of negations. In this section, we show that the expressive power of $\EPSOTC$ collapses to that of existential first-order logic $\EFO$. In this section, $\TC$-operators are applied only to tuples of second-order variables.
As already discussed in the introduction,
this restriction is vital: the formula $[\TC_{x,x'} R(x,x') \lor x=x'](y,y)'$ expresses reachability in directed graphs, which
is not definable even in the full first-order logic.

To facilitate our proofs we start by introducing some helpful terminology.
%
%\subsection{Terminology}
%

\begin{comment}
By $\EPSOTC$ we denote the fragment of $\SO(\STC)$ obtained by the following grammar
\[
\varphi \ddfn  x=y \mid x\neq y \mid R(\vec{x}) \mid \neg R(\vec{x}) \mid \varphi\land\varphi \mid \varphi\lor\varphi \mid 
 \exists x \varphi \mid \exists X \varphi \mid [\STC_{\vec{X}, \vec{X'}} \psi ] \vec{R}, \vec{R'},
\]
where $R$, $R'$, $\vec{R}$, $\vec{R'}$ are relation symbols in $\tau$ and $\psi$ is a formula of $\SO(\STC)(\tau\cup\vec{X}\cup{\vec{X'}})$.
\end{comment}

\begin{definition}
Let  $\vv{a}$ and $\vv{b}$ be tuples  of the same length and $I$ a set of natural numbers.
The \emph{difference} $\dif(\vv{a},\vv{b})$ of the tuples $\vv{a}$ and $\vv{b}$ is defined as follows
\[
\dif(\vv{a},\vv{b}) :=\{ i \mid \vv{a}[i] \neq \vv{b}[i] \}.
\]
The \emph{similarity} $\simi(\vv{a},\vv{b})$ of tuples $\vv{a}$ and $\vv{b}$ is defined as follows
\[
\simi(\vv{a},\vv{b}) :=\{ i \mid \vv{a}[i] = \vv{b}[i] \}.
\]
We say that the tuples $\vv{a}$ and $\vv{b}$ are \emph{pairwise compatible} if the sets $\{ \vv{a}[i] \mid i \in  \mathrm{diff}(\vv{a},\vv{b}) \}$ and $\{ \vv{b}[i] \mid i \in  \mathrm{diff}(\vv{a},\vv{b}) \}$ are disjoint.
The tuples $\vv{a}$ and $\vv{b}$ are \emph{pairwise compatible outside $I$} if $\{ \vv{a}[i] \mid i \in  \mathrm{diff}(\vv{a},\vv{b}), i\notin I \}$ and $\{ \vv{b}[i] \mid i \in  \mathrm{diff}(\vv{a},\vv{b}), i\notin I  \}$ are disjoint.
The tuples $\vv{a}$ and $\vv{b}$ are \emph{pairwise $I$-compatible} if $\vv{a}$ and $\vv{b}$ are pairwise compatible and $\mathrm{sim}(\vv{a},\vv{b}) = I$.
\end{definition}

\begin{definition}
Let $\sigma\subseteq\tau$ be vocabularies, $\mA$ a $\tau$-structure, and $\vv{a}$ a tuple of elements of $A$.
The (quantifier-free) $\sigma$-type of $\vv{a}$ in $\mA$ is the set of those quantifier free $\fo(\sigma)$-formulae $\varphi(\vv{x})$ such that $\mA[\vv{a}/\vv{x}]\models \varphi$.
\end{definition}

%\subsection{Result}

\begin{comment}
The following proposition if proven analogously as \todo[inline]{Cite the paper that simulates FO quantifiers by TC or delete proposition.}

\begin{proposition}
The expressive powers of $\SOTC$, $\SO(\STC)$, $\FO(\TC)$, and $\FO(\STC)$ all coincide.
\end{proposition}
\end{comment}

The following lemma establishes that $\STC$-operators that are applied to $\EFO$-formulas can be equivalently expressed by the finite $\STC^m$-operator.
\begin{lemma}\label{tctotck}
Every formula $\varphi$ of the form $[\TC_{\vec{X}, \vec{X'}} \theta] (\vec{Y}, \vec{Y'})$, where $\theta\in\EFO$ and $\vec{X}$, $\vec{X'}$, $\vec{Y}$, $\vec{Y'}$ are tuples of second-order variables, is equivalent with the formula $[\TC^k_{\vec{X}, \vec{X'}} \theta] (\vec{Y}, \vec{Y'})$, for some $k\in\N$.
\end{lemma}
\begin{proof}
Let $\theta = \exists x_1\dots \exists x_n \psi$, where $\psi$ is quantifier-free, and let $\tau$ denote the vocabulary of $\varphi$.
We will show that for large enough $k$ and for all $\tau$-structures $\mA$
\[
\mA \models [\TC_{\vec{X}, \vec{X'}} \theta] (\vec{Y}, \vec{Y'}) \text{ iff } \mA \models [\TC^k_{\vec{X}, \vec{X'}} \theta] (\vec{Y}, \vec{Y'}).
\]
From here on we consider $\tau$ and $\varphi$ fixed; especially, by a constant, we mean a number that is independent of the model $\mA$; that is, it may depend on $\tau$ and $\varphi$.

 It suffices to show the left-to-right direction as the converse direction holds trivially for all $k$. Assume that $\mA \models  [\TC_{\vec{X}, \vec{X'}} \theta] (\vec{Y}, \vec{Y'})$. By the semantics of $\TC$ there exists a natural number $k_0$ and tuples of relations $\vec{B}_0,\dots, \vec{B}_{k_0}$ on $A$ such that $\vec{B}_0= \vec{Y}^\mA$, $\vec{B}_{k_0}= \vec{Y'}^\mA$, and
 \begin{equation}\label{eq:0}
\mA[\vec{B}_i/\vec{X},  \vec{B}_{i+1}/\vec{X'}]\models \theta, \text{ for } 0\leq i < k_0.
\end{equation}
For each $i< k_0$, let $\mA_i \dfn \mA[\vec{B}_i/\vec{X},  \vec{B}_{i+1}/\vec{X'}]$ and let $\sigma$ denote the vocabulary of $\mA_i$.
By the semantics of the existential quantifier, \eqref{eq:0} is equivalent to saying that 
\begin{equation}\label{eq:1}
\mA_i[\vv{a_i} / x_1,\dots,x_n]\models \psi, \text{ for } 0\leq i < k_0,
\end{equation}
for some $n$-tuples $\vv{a_0},\dots \vv{a_{k_0-1}}$ from $A$.
We will prove the following claim.

\noindent\textbf{Claim:} There exists an index set $I$ and $n+2$ mutually pairwise $I$-compatible sequences in $\vv{a_1},\dots \vv{a_{k_0-1}}$ that have a common $\sigma$-type provided that $k_0$ is a large enough constant.

\noindent\textbf{Proof of the claim:} Let $\vv{c}^0=(\vv{c_0}^0,\vv{c_2}^0,\dots,\vv{c_t}^0)$ denote the longest (not necessarily consecutive) subsequence  of $\vv{a_1},\dots \vv{a_{k_0-1}}$ that have a common $\sigma$-type. Since there are only finitely many $\sigma$-types, $t$ can be made as large as needed by making $k_0$ a large enough constant. 

We will next show that there exists $n+2$ mutually pairwise $I$-compatible sequences in $\vv{c}^0$ for some $I$ (provided that $t$ is large enough). Set $\SIM_0:=\emptyset$. In the construction below we maintain the following properties for $0\leq i \leq n$:
\begin{itemize}
\item For each $j\in \SIM_i$ and for each tuple $\vv{a}$ and $\vv{b}$ in $\vv{c}^i$ it holds that $\vv{a}[j]=\vv{b}[j]$.
\item The length of $\vv{c}^i$ is as long a constant as we want it to be.
\end{itemize}
For $l < n$, let $\vv{b}^l_0,\dots,\vv{b}^l_{t_l}$ be a maximal collection (in length) mutually pairwise $\mathrm{SIM}_l$-compatible sequences from $\vv{c}^l$. If $t_l\geq n+1$ we are done. Otherwise note that, since each $\vv{b}^l_j$ is an $n$-tuple, the number of different points that may occur in $\vv{b}^l_0,\dots,\vv{b}^l_{t_l}$ is $\leq n^2+n$.
By an inductive argument we may assume that the length of $\vv{c}^{l}$ is as large a constant as we want, and thus we may conclude that there exists an index $i\notin \SIM_l$ and an element $d_l$ such that there are as many as we want tuples $\vv{c}^l_j$ in  $\vv{c}^l$ such that $\vv{c}^l_j[i]=d_l$. Set $\SIM_{l+1} := \SIM_l \cup \{ i \}$ and let $\vv{c}^{l+1}$ be the sequence of exactly those  $\vv{a}\in\vv{c}^l$ such that $\vv{a}[i]=d_l$. Notice that the length of $\vv{c}^{l+1}$ is as large a constant as we want it to be.

Finally, the case $l=n$. Note that $\mathrm{SIM}_n=\{0,\dots,n-1\}$ and $\vv{c}^n$ is a sequence of $n$-tuples; in fact all tuples in $\vv{c}^n$ are identical.  Thus, if the length of $\vv{c}^n$ is at least $n+2$, the first $n+2$ sequences of  $\vv{c}^n$  constitute a mutually pairwise $\mathrm{SIM}_n$-compatible sequence of length $n+2$. It is now straightforward but tedious to check how large $k_0$ has to be so that the length of $\vv{c}^n$ is at least $n+2$; thus the claim holds. \qed

Now let $\vv{a}_{i_0},\dots,\vv{a}_{i_{n+1}}$, $0 < i_0 <\dots <i_{n+1}$, be mutually pairwise $I$-compatible sequences from $\vv{a}_1,\dots \vv{a}_{k_0-1}$ with a common $\sigma$-type provided by the Claim.
Let $1\leq j \leq n+1$ be an index such that $\vv{a}_{i_0-1}$ and $\vv{a}_{i_j}$ are pairwise compatible outside $I$ and $\simi(\vv{a}_{i_0-1}, \vv{a}_{i_j})\subseteq I$.
It is straightforward to check that such a $j$ always exists, for if $\vv{a}_{i_0-1}$ and $\vv{a}_{i_{j'}}$ are not pairwise compatible outside $I$ or $\simi(\vv{a}_{i_0-1}, \vv{a}_{i_j'})\not\subseteq I$, there exists some indices ${m},{m'}\notin I$ such that $\vv{a}_{i_0-1}[m]=\vv{a}_{i_{j'}}[m']$, and for each such $\vv{a}_{i_{j'}}$ the value of the related $\vv{a}_{i_{j'}}[m']$ has to be unique as  $\vv{a}_{i_1},\dots,\vv{a}_{i_{n+1}}$ are mutually pairwise $I$-compatible. Now $j$ must exist since the length of $\vv{a}_{i_1},\dots,\vv{a}_{i_{n+1}}$ is $n+1$ while the length of $\vv{a}_{i_0-1}$ is only $n$.

Consider the models $\mA_{i_{0}-1} = \mA[\vec{B}_{i_{0}-1}/ \vec{X}, \vec{B}_{i_{0}}/ \vec{X'}]$ and $\mA_{i_{j}} = \mA[\vec{B}_{i_{j}}/ \vec{X}, \vec{B}_{i_{j+i}}/ \vec{X'}]$ and recall that
\(
\mA_{i_{0}-1}[\vv{a}_{i_0-1} / x_1,\dots, x_n] \models \psi \text{ and } \mA_{i_{j}}[\vv{a}_{i_j} / x_1,\dots, x_n]\models \psi.
\)
We claim that there exists a sequence $\vec{B}$ of relations on $A$ such that
\begin{equation}\label{modelswap}
\mA[\vec{B}_{i_{0}-1}/ \vec{X}, \vec{B}/ \vec{X'}, \vv{a}_{i_0-1} / x_1,\dots, x_n] \models \psi \text{ and } 
\mA[\vec{B}/ \vec{X}, \vec{B}_{i_{j+i}}/ \vec{X'}, \vv{a}_{i_j} / x_1,\dots, x_n]\models \psi.
\end{equation}
and thus that $\mA[\vec{B}_{i_{0}-1}/ \vec{X}, \vec{B}/ \vec{X'}]\models \theta$ and $\mA[\vec{B}/ \vec{X}, \vec{B}_{i_{j+i}}/ \vec{X'}]\models \theta$.
From this the claim of the theorem follows for $k=k_0$.

It now suffices to show that such a $\vec{B}$ exists. The idea is that $\vec{B}$ looks exactly like $\vec{B}_{i_{0}}$ with respect to points in $\vv{a}_{i_0-1}$ and like $\vec{B}_{i_{j}}$ with respect to points  $\vv{a}_{i_j}$. Formally $\vec{B}$ is defined as follows.
For every relation $\vec{B}[m]$ and tuple $\vv{a}\in A^{\ar{\vec{B}[m]}}$
\begin{itemize}
\item if $\vv{a}$ is completely included in neither $\vv{a}_{i_0-1}$ nor $\vv{a}_{i_j}$ then we set $\vv{a}\notin \vec{B}[m]$,
\item if $\vv{a}$ is completely included in $\vv{a}_{i_0-1}$ then we set $\vv{a}\in \vec{B}[m] $ iff $\vv{a}\in \vec{B}_{i_0}[m]$,
\item if $\vv{a}$ is completely included in $\vv{a}_{i_j}$ then we set $\vv{a}\in \vec{B}[m]$ iff $\vv{a}\in \vec{B}_{i_j}[m]$.
\end{itemize}
Note that if $\vv{a}=(a_1,\dots,a_m)$ is completely included in both $\vv{a}_{i_0-1}$ and $\vv{a}_{i_j}$ then there exists indices $j_1,\dots j_m \in I$ such that, for $1\leq l \leq m$, $a_l=\vv{a}_{i_j}[j_l]=\vv{a}_{i_0}[j_l]$. The former equality follows, with indices in $I$, since $\vv{a}_{i_0-1}$ and $\vv{a}_{i_j}$ are pairwise compatible outside $I$ and $\simi(\vv{a}_{i_0-1}, \vv{a}_{i_j})\subseteq I$. The latter equality follows since $\vv{a}_{i_0}$ and $\vv{a}_{i_j}$ are pairwise $I$-compatible. Since $\vv{a}_{i_0}$ and $\vv{a}_{i_j}$ have the same $\sigma$-type $\vv{a}\in \vec{B}_{i_0}[m]$ iff $\vv{a}\in \vec{B}_{i_j}[m]$, for all $m$, and thus $\vec{B}$ is well-defined. It is now immediate that \eqref{modelswap} holds.
\end{proof}

\begin{lemma}\label{inductionbasis}
For every formula of vocabulary $\tau$ of the form $\exists X \theta$ or $[\TC_{\vec{X}, \vec{X'}} \theta] (\vec{Y}, \vec{Y'})$, where $\theta\in\EFO$ and $\vec{X}$, $\vec{X'}$, $\vec{Y}$, $\vec{Y'}$ are tuples of relation variables, there exists an equivalent formula $\varphi\in\EFO$ of vocabulary $\tau$.  
\end{lemma}

\begin{proof}
Consider first the formula $\exists X \theta$. Define $n\dfn \ar{X}$ and let $k$ be the number of occurrences of $X$ in $\theta$. The idea behind our translation is that the quantification of $X$ can be equivalently replaced by a quantification of an $n$-ary relation of size $\leq k$; this can be then expressed in $\EFO$ by quantifying $k$ many $n$-tuples (content of the finite relation).

Let $\theta_\emptyset$ denote the formula obtained from $\theta$ by replacing every occurrence of the relation variable $X$ of the form $X(\vec{x})$ in $\theta$ by the formula $\exists x (x\neq x)$. Define
\(
\gamma \dfn \exists \vv{x}_1 \dots \exists \vv{x}_k  (\theta_\emptyset \lor \theta'),
\)
where, for each i, $\exists \vv{x}_i$ is a shorthand for $\exists {x}_{1,i}\dots \exists {x}_{n,i}$ and $\theta'$ is the formula obtained from $\theta$ by substituting each occurrence of the relation variable $X$ of the form $X(\vv{x})$ in $\theta$ by %the formula
\(
\bigvee_{1\leq i\leq n} (\vv{x}= \vv{x}_i).
\)
It is straightforward to check that $\gamma$ is an $\EFO$-formula of vocabulary $\tau$ equivalent with $\exists X \theta$.

Consider then the formula $\varphi=[\TC_{\vec{X}, \vec{X'}} \theta] (\vec{Y}, \vec{Y'})$. In order to simplify the presentation, we stipulate that $\vec{X}$ and $\vec{X'}$ are of length one, that is, variables $X$ and $X'$, respectively; the generalisation of the proof for arbitrary tuples of second-order variables is straightforward.
By Lemma \ref{tctotck}, we obtain $k\in\N$ such that $\varphi$ and $\varphi' \dfn [\TC^k_{X, X'} \theta] (Y, Y')$ are equivalent.

The following formulas are defined via substitution; by $\theta(A/B)$ we denote the formula obtained from $\theta$ by substituting each occurrence of the symbol $B$ by the symbol $A$. 
\begin{itemize}
\item $\theta^{\mathrm{end}}_0 \dfn \theta(Y/X, Y'/X')$ and $\theta^{\mathrm{end}}_i \dfn \theta(X_i/X, Y'/X')$, for  $1 \leq i < k$,
%\item for  $1 \leq i < k$, define $\theta^{\mathrm{end}}_i \dfn \theta(X_i/X, Y'/X')$,
\item $\theta^{\mathrm{move}}_1 \dfn \theta(Y/X, X_1/X')$ and $\theta^{\mathrm{move}}_i  \dfn \theta(X_{i-1}/X, X_{i}/X')$, for $2 \leq i < k$.
%\item for $2 \leq i < k$, define $\theta^{\mathrm{move}}_i  \dfn \theta(X_{i-1}/X, X_{i}/X') $
\end{itemize}
Let $\psi$ denote the following formula of existential second-order logic
\[
\exists X_1 \dots \exists X_{k-1} \; \bigvee_{0\leq n < k} ( \theta^{\mathrm{end}}_n \land \bigwedge_{1\leq i \leq n} \theta^{\mathrm{move}}_i ).
\]
It is immediate that $\varphi'$ and $\psi$ are equivalent. Note that $\psi$ is of the form $\exists X_1 \dots \exists X_{k-1} \psi'$, where $\psi'$ is an $\EFO$-formula. By repetitively applying the first case of this lemma to subformulas of $\psi$, we eventually obtain an equivalent $\EFO$-formula over $\tau$ as required.
\end{proof}

\begin{theorem}
The expressive powers of $\EPSOTC$ and $\EFO$ coincide.
\end{theorem}
\begin{proof}
Every $\EFO$-formula is also an $\EPSOTC$-formula, and thus it suffices to establish the converse direction. The proof proceeds via induction on the combined nesting depth of second-order quantifiers and $\TC$ operators. For every formula $\varphi$ of the form $\exists X \theta$ or $[\TC_{\vec{X}, \vec{X'}} \theta] (\vec{Y}, \vec{Y'})$, where $\theta\in\EFO$, let $\varphi^*$ denote the equivalent $\EFO$-formula obtained from Lemma \ref{inductionbasis}.

Let $\psi$ be an arbitrary $\EPSOTC$-formula with combined nesting depth $k$ of second-order quantifiers and $\TC$ operators. Let $\psi'$ be the $\EPSOTC$-formula obtained from $\psi$ by simultaneously substituting each of its subformula $\varphi$ of the form $\exists X \theta$ or $[\TC_{\vec{X}, \vec{X'}} \theta] (\vec{Y}, \vec{Y'})$, where $\theta\in\EFO$, by its $\EFO$ translation $\varphi^*$. Clearly $\psi$ and $\psi'$ are equivalent, and the combined nesting depth of second-order quantifiers and $\TC$ operators in $\psi'$ is $k-1$. Thus, by induction, the claim follows, as for $k=0$ the formula is already in $\EFO$.
\end{proof}

\section{MSO(TC) and counting}\label{sec:counting}
We define a counting extension of $\MSO(\TC)$ and show that the extension does
not add  expressive power to the logic.
In this way, we demonstrate that quite a bit of queries involving
counting can be expressed already in $\MSOTC$.

\subsection{Syntax and semantics of CMSO(TC)} \label{sec:syntaxcounting}
We assume a sufficient supply of \emph{counter variables} or simply  \emph{counters}, which are a new
sort of variables. We use the Greek letters $\mu$ and $\nu$ (with subscripts) to denote counter variables.
The notion of a vocabulary is extended so that
it may also contain counters.  A structure $\mA$ over a
vocabulary $\tau$ is defined to be a pair $(A,I)$ as before, where
$I$ now also maps the
counters in $\tau$ to elements of $\{0,\dots,n\}$,
where $n$ is the cardinality of $A$.

We also assume a sufficient supply of \emph{numeric predicates}. Intuitively numeric predicates are relations over natural numbers such as the tables of multiplication and addition. Technically, we use an approach similar to generalised quantifiers; a $k$-ary numeric predicate is a class $Q_p\subseteq \N^{k+1}$ of $k+1$-tuples of natural numbers. For a numeric predicate $Q_p$, we use $p$ as a symbol referring to the predicate. For simplicity, we often call $p$ also numeric predicate. Note that when evaluating a $k$-ary numeric predicate $p(\mu_1,\dots,\mu_k)$ on a finite structure $\mA$, we let the numeric predicate $Q_p$ access also the cardinality of the structure in question, and thus $Q_p$ consists of $k+1$-tuples and not $k$-tuples. This convention allows us, for example, to regard the modular sum $a+b \equiv c \,(\mathrm{mod}\, n)$, where $n$ refers to the cardinality of the structure, as a $3$-ary numeric predicate.

We consider only those numeric predicates which can be decided in $\NL$.
Since, on finite ordered structures, first-order transitive closure logic captures $\NL$, this boils down to being definable in first-order transitive closure logic when the counter variables are interpreted as points in an ordered structure representing an initial segment of natural numbers (see Definition \ref{def:NLpredicate} and Proposition \ref{prop:defining} below for precise formulations).

\begin{definition}
The syntax of $\CMSOTC$ extends the syntax of $\MSOTC$ as follows:
\begin{itemize}
\item
Let $\varphi$ be a formula, $\mu$ a counter, and $x$
a first-order variable.  Then $\mu = \#\{x \mid \varphi\}$ is also
a formula.  The set of its free variables is defined to be
$(\FV(\varphi) - \{x\}) \cup \{\mu\}$.
\item
If $\varphi$ is a formula and $\mu$ a counter then also $\exists
\mu \, \varphi$ is a formula with set of free variables
$\FV(\varphi) - \{\mu\}$.
\item
Let $\mu_1$, \dots, $\mu_k$ be counters and let $p$ be a $k$-ary numeric
predicate.  Then $p(\mu_1,\dots,\mu_k)$ is a formula with the set of free
variables $\{\mu_1,\dots,\mu_k\}$.
\item
  The scope of the transitive-closure operator is widened to
    apply as well to counters.  Formally, in a formula of the
    form $[\TC_{\vec{X},\vec{X'}}\varphi](\vec{Y},\vec{Y'})$, the
    variables in $\vec{X}$, $\vec{X'}$, $\vec{Y}$, and $\vec{Y'}$ may also include
    counters.  We still require that the tuples
    $\vec{X}$, $\vec{X'}$, $\vec{Y}$, and $\vec{Y'}$
    have the same sort, i.e., if a counter appears in some
    position in one of these tuples then a counter must appear in
    that position in each of the tuples.
\end{itemize}
\end{definition}
\begin{definition}
The satisfaction relation, $\mA\models \psi$, for $\CMSOTC$ formulas $\psi$ and structures $\mA=(A,I)$ over a vocabulary appropriate for $\psi$ is defined in the same way as for $\MSOTC$ with the following additional clauses.
\begin{itemize}
\item
Let $\psi$ be of the form $\exists \mu \, \varphi$, where $\mu$ is a
counter, and  let $n$ denote the cardinality of $A$.
Then $\mA \models \psi$ iff there exists a number $i \in
\{0,\dots,n\}$ such that $\mA[i/\mu] \models \varphi$.
\item
Let $\psi$ be of the form $\mu = \#\{x \mid \varphi\}$.
Then $\mA \models \psi$ iff $I(\mu)$ equals the
cardinality of the set $ \{a \in A \mid \mA[a/x]
\models \varphi\}.$
\item 
Let $\psi$ be of the form $p(\mu_1,\dots,\mu_k)$, where
$\mu_1$, \dots, $\mu_k$ are counters and $p$ is a $k$-ary numeric
predicate. Then
\(
\mA \models p(\mu_1,\dots,\mu_k) \text{ iff }   \big(\lvert A\rvert, I(\mu_1),\dots,I(\mu_k)\big) \in Q_p.
\)
\end{itemize}
\end{definition}

\begin{definition}\label{def:NLpredicate}
A $k$-ary numeric predicate $Q_p$ is \emph{decidable in $\NL$} if the
membership
$(n_0,\dots,n_k)\in Q_p$ can be decided by a nondeterministic
Turing machine that uses logarithmic space when the numbers
$n_0,\dots,n_k$ are given in unary. Note that this is equivalent
to linear space when $n_0,\dots,n_k$ are given in binary.
\end{definition}
From now on we restrict our attention to  numeric predicates that are decidable in $\NL$.
The following proposition follows directly from a result of Immerman (Theorem \ref{thm:fotc}) that, on ordered structures,
$\FO(\FTC)$ captures $\NL$.
\begin{proposition}\label{prop:defining}
For every $k$-ary numeric predicate $Q_p$ decidable in $\NL$ there exists a formula $\varphi_p$ of $\FO(\FTC)$ over $\{s, x_1,\dots,x_k\}$, where $s$ is a binary second-order variable and $x_1,\dots,x_k$ are first-order variables, s.t. for all appropriate structures $\mA$ for $p(\mu_1,\dots,\mu_k)$
\[
\mA \models p(\mu_1,\dots,\mu_k) \text{ iff }    \big(\lvert A\rvert, I(\mu_1),\dots,I(\mu_k)\big) \in Q_p  \text{ iff } (B,J)\models \varphi_p,  
\]
where $B=\{0,1,\dots, \lvert A \rvert\}$, $J(s)$ is the successor relation of $B$, and $J(x_i)=I(\mu_i)$, for $1\leq i \leq k$. 
\end{proposition}

%\todo[inline]{Jonni:The above could also be written as $(n_0,\dots, n_k) \in Q_p$ iff $(B,J)\models \varphi_p$}

%\begin{definition}
%For each $\FO(\FTC)$ formula $\varphi$ of vocabulary $\{s, x_1,\dots,x_k\}$, where $s$ is a binary second-order variable, $x_1,\dots,x_k$ first-order variables, and $k\in\N$, we identify a $k$-ary counter predicate $p_{\varphi}$. We say that  $p_{\varphi}$ is an  \emph{$\FO(\FTC)$-definable counter predicate} and that $\varphi$ is the \emph{defining formula} of $p_\varphi$
%\end{definition}

\subsection{CMSO(TC) collapses to MSO(TC)}

Let $\tau$ be a vocabulary with counters. Let $\tau^*$ denote
the vocabulary without counters obtained from $\tau$ by viewing
  each counter variable of $\tau$ as a set variable.  Let $\mA=(A,I)$ be a structure over
  $\tau$, and let $\mB=(A,J)$ be a structure over $\tau^*$ with
  the same domain as $\mA$.  We say
  that $\mB$ \emph{simulates} $\mA$ if for every counter $\mu$ in
  $\tau$, the set $J(\mu)$ has cardinality $I(\mu)$, and $J(X)=I(X)$, for each first-order or second-order variable $X\in \tau$.
  Let $\varphi$ be a $\CMSOTC$-formula over $\tau$ and
  $\psi$ an $\MSOTC$ formula over $\tau^*$.  We say that
  $\psi$ \emph{simulates} $\varphi$ if whenever $\mB$ simulates
  $\mA$, we have that  $\mA \models \varphi$ if and only if $\mB \models \psi$.
  %, we have 
  %$$ \mA \models \varphi \quad \Leftrightarrow
  %\quad \mB \models \psi. $$
%
%

Let $\varphi(x)$ and $\psi(y)$ be formulae of some logic. The \emph{H\"artig quantifier} is defined as follows:
\begin{align*}
\mA \models \Ha xy(\varphi(x),\psi(y)) \Leftrightarrow\; &\text{the sets $\{a \in A \mid \mA[a/x]\models \varphi\}$ and $\{b \in A \mid \mA[b/y]\models \psi\}$} \\
& \text{have the same cardinality }
\end{align*}
\begin{proposition}\label{prop:hartig}
The H\"artig quantifier  can be expressed in $\MSOTC$.
\end{proposition}
\begin{proof}
Consider a structure $(A,I)$ and monadic second-order variables $X$, $Y$, $X'$ and $Y'$.
Let $\psi_{\mathrm{decrement}}$ denote an $\FO$-formula expressing that $I(X') = I(X)\setminus \{a\}$ and $I(Y') = I(Y) \setminus \{b\}$, for some $a\in I(X)$ and $b\in I(Y)$. Define
  \[
\psi_{\mathrm{ec}} \dfn  \exists X_{\emptyset}\Big(  \big( \forall x  \neg X_{\emptyset}(x) \big) \land [\TC_{X,Y,X',Y'} \psi_{\mathrm{decrement}}] (Z, Z',X_{\emptyset},X_{\emptyset}) \Big).
  \]
It is straightforward to check that $\psi_{\mathrm{ec}}$ holds in $(A,I)$ if and only if $\lvert I(Z) \rvert = \lvert I(Z') \rvert$. Therefore $\Ha xy(\varphi(x),\psi(y))$ is equivalent with the formula
\[
\exists Z \exists Z'  \big( \forall x (\varphi(x) \leftrightarrow Z(x)) \land  \forall y (\psi(y) \leftrightarrow Z'(y)) \land \psi_{\mathrm{ec}} \big),
\]
assuming that $Z$, $Z'$ are variable symbols that occur in neither $\varphi$ nor $\psi$.
\end{proof}

%\todo[inline]{Jonni: We could mention also other quantifiers such as majority.}

\begin{lemma}\label{lemma:trans}
Let $\tau=\{s, x_1,\dots,x_n\}$ and $\sigma=\{X_1,\dots,X_n\}$ be vocabularies, where $s$ is a binary second-order variable, $x_1,\dots,x_n$ are first-order variables, and $X_1,\dots,X_n$ are monadic second-order variables. For every $\FO(\FTC)$-formula $\varphi$ over $\tau$ there exists an $\MSOTC$-formula $\varphi^+$ over $\sigma$ such that
\[
(A,I)\models \varphi  \quad\Leftrightarrow\quad (B,J)\models \varphi^+,
\]
for every $(A,I)$ and $(B,J)$ such that
$(A,I)$ is a structure over vocabulary $\tau$, where $A=\{0,\dots,m\}$, for some $m\in\N$, and $I(s)$ is the canonical successor relation on $A$, and $(B,J)$ is a structure over vocabulary $\sigma$ such that $\lvert B\rvert = m$ and $\vert J(X_i) \rvert = I(x_i)$, for $1\leq i \leq n$.
\end{lemma}
\begin{proof}
We define the translation ${}^+$ recursively as follows. In the translation, we introduce for each first-order variable $x_i$ a monadic second-order variable $X_i$ by using the corresponding capital letter with the same index. Consequently, in tuples of variables, identities between the variables are maintained. The idea of the translation is that natural numbers $i$ are simulated by sets of cardinality $i$. Identities between first-order variables are then simulated with the help of the H\"artig quantifier, which, by Proposition \ref{prop:hartig}, is definable in $\MSOTC$.
\begin{itemize}
\item For $\psi$ of the form $x_i=x_j$, define $\psi^+ \dfn \Ha xy \big(X_i(x), X_j(y)\big)$.
\item For $\psi$ of the form $s(x_i,x_j)$, define $\psi^+ \dfn \exists z \Big( \neg X_i(z) \land \Ha xy \big(X_i(x)\lor x=z, X_j(y)\big)\Big)$.
\item For $\psi$ of the form $\neg \varphi$ and $(\varphi\land \theta)$, define $\psi^+$ as $\neg\varphi^+$ and $(\varphi^+\land \theta^+)$, respectively.
\item For $\psi$ of the form $\exists x_i \varphi$, define $\psi^+ \dfn \exists X_i \varphi^+$, where $X_i$ is a monadic second-order variable.
\item For $\psi$ of the form $[\TC_{\vec{x},\vec{x'}} \varphi] (\vec{y},\vec{y'})$, define $\psi^+ \dfn [\TC_{\vec{X},\vec{X'}} \varphi^+] (\vec{Y},\vec{Y'})$, where $\vec{X}$, $\vec{X'}$, $\vec{Y}$, and $\vec{Y'}$ are tuples of monadic second-order variables that correspond to the tuples $\vec{x}$, $\vec{x'}$, $\vec{y}$, and $\vec{y'}$ of first-order variables.
\end{itemize}
The correctness of the translation follows by a simple inductive argument.
\end{proof}

With the help of the previous lemma, we are now ready to show how $\CMSOTC$-formulas can be simulated in $\MSOTC$.
  \begin{theorem}\label{thm:simulation}
    Every $\CMSOTC$-formula can be simulated by an $\MSOTC$-formula.
  \end{theorem}
  \begin{proof}
  Let $\tau$ be a vocabulary with counters and $\tau^*$ the vocabulary without counters obtained from $\tau$ by viewing
  each counter as a set variable. We define recursively a translation ${}^*$ that maps $\CMSOTC$-formulas over vocabulary $\tau$ to $\MSOTC$-formulas over $\tau^*$.
 \begin{itemize}
\item For $\psi$ of the form $x_i=x_j$, define $\psi^* \dfn x_i=x_j$.
\item For $\psi$ of the form $X(x_1,\dots,x_n)$, define $\psi^* \dfn X(x_1,\dots,x_n)$.
\item For an $\NL$ numeric predicate  $Q_p$ and $\psi$ be of the form $p(\mu_1,\dots,\mu_k)$, define $\psi^*$ as $\varphi_p^+(\mu_1/X_1,\dots,\mu_k/X_k)$, where ${}^+$ is the translation defined in Lemma \ref{lemma:trans} and $\varphi_p$ the defining formula of $Q_p$ obtained from Proposition \ref{prop:defining}.
\item  For $\psi$ of the form $\mu = \#\{x \mid \varphi\}$, define $\psi^*$ as the $\MSOTC$-formula $\Ha xy(\varphi^*,\mu(y))$.
\item For $\psi$ of the form $\neg \varphi$ and $(\varphi\land \theta)$, define $\psi^*$ as $\neg\varphi^*$ and $(\varphi^*\land \theta^*)$, respectively.
\item For $\psi$ of the form $\exists x_i \varphi$, $\exists \mu_i \varphi$, and $\exists X_i \varphi$, define $\psi^*$ as $\exists x_i \varphi^*$, $\exists \mu_i \varphi^*$, and $\exists X_i \varphi^*$. Remember that, on the right,  $\mu_i$ is treated a as a monadic second-order variable.
\item For $\psi$ of the form $[\TC_{\vec{X},\vec{X'}} \varphi] (\vec{Y},\vec{Y'})$, define $\psi^* \dfn [\TC_{\vec{X},\vec{X'}} \varphi^*] (\vec{Y},\vec{Y'})$.
\end{itemize}
We claim that, for every $\CMSOTC$-formula $\psi$ over $\tau$, $\psi^*$ is an $\MSOTC$-formula over $\tau^*$ that simulates $\psi$. 
Correctness of the simulation follows by induction using Lemma \ref{lemma:trans} and Proposition \ref{prop:defining}.

We show the case for the numeric predicates. Let $\mA=(A,I)$ be a $\tau$-structure and $\mA^*$ a $\tau^*$-structure that simulates $\mA$. Let $Q_p$ be a $k$-ary $\NL$ numeric predicate, $\mu_1,\dots,\mu_k$ counters from $\tau$, and $\varphi_p$ the defining $\FO(\FTC)$-formula of $Q_p$ given by Proposition \ref{prop:defining}. Then, by Proposition \ref{prop:defining},
\[
\mA\models p(\mu_1,\dots,\mu_k) \text{ iff } (B,J)\models \varphi_p,  
\]
where $B=\{0,1,\dots, \lvert A \rvert\}$, $J(s)$ is the successor relation of $B$, and $J(x_i)=I(\mu_i)$, for $1\leq i \leq k$. Let ${^+}$ denote the translation from $\FO(\FTC)$ to $\MSOTC$ defined in Lemma \ref{lemma:trans}. Then, by Lemma \ref{lemma:trans}, it follows that $(B,J)\models \varphi_p \text{ iff } \mA\models \varphi_p^+$.
\end{proof}
In the next example, we introduce notation for some $\MSOTC$-definable numeric predicates that are used in the following sections. 
\begin{example}\label{ex:predicates}
Let $k$ be a natural number, $X,Y,Z,X_1,\dots,X_n$ monadic second-order variables, and $\mA=(A,I)$ an appropriate structure. The following numeric predicates are clearly $\NL$-definable and thus, by Theorem \ref{thm:simulation}, definable in $\MSOTC$:
\begin{itemize}
\item $\mA\models\size(X,k)$ iff $\lvert I(X) \rvert = k$,
\item $\mA\models \times(X,Y,Z)$ iff $\lvert I(X) \rvert \times \lvert I(Y) \rvert = \lvert I(Z) \rvert$,
\item $\mA\models +(X_1,\dots, X_n,Y)$ iff $\lvert I(X_1) \rvert +\dots+ \lvert I(X_n) \rvert = \lvert I(Y) \rvert$.
\end{itemize}
\end{example}

\section{Order-invariant MSO}\label{sec:order}
Order-invariance plays an important role in finite model theory. In descriptive complexity theory many characterisation rely on the existence of a linear order. However the particular order in a given stricture is often not important. Related to applications in computer science, it is often possible to access an ordering of the structure that is not controllable and thus a use of the ordering should be such that change in the ordering should not make a difference. Consequently, in both cases order can be used, but in a way that the described properties are \emph{order-invariant}.

Let $\tau_\leq \dfn \tau\cup\{\leq\}$ be a finite vocabulary, where $\leq$ is a binary relation symbol. A formula $\varphi\in\MSO$ over $\tau_\leq$ is \emph{order-invariant}, if for every $\tau$-structure $\mA$ and expansions $\mA'$ and $\mA^*$ of $\mA$ to the vocabulary $\tau_\leq$, in which $\leq^{\mA'}$ and $\leq^{\mA^*}$ are total linear orders of $A$, we have that 
\(
\mA' \models \varphi \text{ if and only if } \mA^*\models \varphi.
\)
A class $\mathcal{C}$ of $\tau$-structures is \emph{definable in order-invariant $\MSO$} if and only if the class
\(
\{(\mA,\leq )  \mid \mA\in \mathcal{C} \text{ and $\leq$ is a complete linear order of $A$}\}
\)
is definable by some order-invariant $\MSO$-formula.

We call a vocabulary $\tau$ \emph{a unary vocabulary} if it consists of only monadic second-order variables.
In this section we establish that on unary vocabularies $\MSO(\TC)$ is strictly more expressive than order-invariant $\MSO$. The separation holds already for the empty vocabulary. 

\subsection{Separation on empty vocabulary}
First note that over vocabulary $\{\leq\}$ there exists only one structure, up to isomorphism, of size $k$, for each $k\in N$, in which $\leq$ is interpreted as a complete linear order. Consequently, every $\MSO$-formula of vocabulary $\{\leq\}$ is order-invariant. Also note that, in fact, $\{\leq\}$-structures interpreted as word models correspond to finite strings over some fixed unary alphabet. Thus, via B\"uchi's theorem, we obtain that, over the empty vocabulary, order-invariant $\MSO$ captures essentially regular languages over unary alphabets.

\begin{definition}
For a finite word $w$ of some finite alphabet $\Sigma=\{a_1\,\dots,a_k\}$, a \emph{Parikh vector} $p(w)$ of $w$ is the $k$-tuple $(\vert w\rvert_{a_1},\dots.\vert w\rvert_{a_k})$ where $\vert w\rvert_{a_i}$ denotes the number of $a_i$s in $w$. A \emph{Parikh image} $P(L)$ of a language $L$ is the set $\{p(w) \mid w\in L\}$ of Parikh vectors of the words in the language.
\end{definition}

A subset $S$ of $\N^k$ is a \emph{linear set} if  
\(
S= \{\vec{v}_0  + \sum_{i=1}^m  a_i \vec{v}_i \mid a_1,\dots,a_m\in\N\}
\)
for some \emph{offset} $\vec{v}_0\in \N^k$ and \emph{generators} $\vec{v}_1,\dots,\vec{v}_m\in \N^k$.
$S$ is \emph{semilinear} if it is a finite union of linear sets.

\begin{theorem}[Parikh's theorem, \cite{Parikh:1966}]
For every regular language $L$ its Parikh image $P(L)$ is a finite union of linear sets.
\end{theorem}
We use the following improved version of Parikh's theorem:
\begin{theorem}[\cite{KoTo10}]\label{Parikhthm}
For every regular language $L$ over alphabet of size $k$ its Parikh image $P(L)$ is a finite union of linear sets with at most $k$ generators.
\end{theorem}
Especially, the Parikh image of a regular language $L$ over the unary alphabet can be described as a finite union of sets of the form $\{a + n b \mid n\in\N\}$, where $a,b\in\N$. 

\begin{proposition}\label{prop:prime}
The class $\mathcal{C} = \{\mA \mid \lvert A \rvert \text{ is a prime number}\}$ of $\emptyset$-structures is not definable in order-invariant $\MSO$.
\end{proposition}
\begin{proof}
First realise that the set of prime numbers $\primes$ is not semilinear. Consequently, it follows from Parikh's theorem that any language that is regular cannot have $\primes$ as its Parikh image. Thus, by B\"uchi's theorem, the class of word models $\mB$ over vocabulary $\{\leq, P_a\}$ such that $\lvert B\rvert \in\primes$ is not definable in $\MSO$. Hence $\mathcal{C}$ is not definable in order-invariant $\MSO$. 
\end{proof}
However the following example establishes the class $\mathcal{C}$ is definable in $\MSO(\TC)$.

\begin{example}\label{ex:prime}
The class $\mathcal{C} = \{\mA \mid \lvert A \rvert \text{ is a
prime number}\}$ of $\emptyset$-structures is defined by the
following formula in
$\MSO(\TC)$. We use $\MSOTC$-definable numeric predicates introduced in Example \ref{ex:predicates}.
\[
\exists X \forall Y \forall Z\big(  \forall x (X(x)) \land (\size(Y,1) \lor \size(Z,1) \lor \neg \times(Y,Z,X)) \big).
\]
\end{example}

\begin{corollary}\label{cor:separation}
For any vocabulary $\tau$, there exists a class $\mathcal{C}$ of $\tau$-structures such that $\mathcal{C}$ is definable in $\MSOTC$ but it is not definable in order-invariant $\MSO$.
\end{corollary}

\subsection{Inclusion on unary vocabularies}

We will show that every class of structures over a unary vocabulary $\tau$ that is definable in order-invariant $\MSO$ is also definable in $\MSO(\TC)$.

\begin{definition}
Let $\tau= \{X_1,\dots,X_k\}$ be a finite unary vocabulary and let $Y_1,\dots,Y_{2^k}$ denote the Boolean combinations of the variables in $\tau$ in some fixed order. For every structure $\mA=(A,I)$ over $\tau$, we extend the scope of $I$ to include also $Y_1,\dots,Y_{2^k}$ in the obvious manner. The \emph{Parikh vector $p(\mA)$ of $\mA$} is the $2^k$-tuple $\big(\lvert I(Y_1) \rvert ,\dots, \lvert I(Y_{2^k}) \rvert \big)$.  A \emph{Parikh image $P(\mathcal{C})$ of a class of $\tau$-structures $\mathcal{C}$} is the set $\{p(\mA) \mid \mA \in \mathcal{C}\}$.
\end{definition}

\begin{comment}
The proof of the following lemma is immediate.
\begin{lemma}
Let $\tau$ be a finite unary vocabulary and let $\mathcal{C}$ be an order-invariant $\MSO$ definable class of $\tau$-structures. Then $\mathcal{C}$ is invariant under Parikh vectors, i.e., if $p(\mA)=p(\mA')$ then $\mA\in\mathcal{C}$ if and only if $\mA'\in\mathcal{C}$.
\end{lemma}
The previous lemma establishes that order-invariant $\MSO$ definable classes are characterised by their Parikh images.

\begin{lemma}
Let $\tau$ be a finite unary vocabulary and let $\mathcal{C}$ be an order-invariant $\MSO$ definable class of $\tau$-structures. There exists a regular language $L$ such that the Parikh image of $L$ is the Parikh image of $\mathcal{C}$.
\end{lemma}
\end{comment}

\begin{theorem}
Over finite unary vocabularies $\MSOTC$ is strictly more expressive than order-invariant $\MSO$.
\end{theorem}
\begin{proof}
Strictness follows directly from Corollary \ref{cor:separation} and thus it suffices to establish inclusion.
Let $\tau=\{X_1,\dots,X_k\}$ be a finite unary vocabulary and $\varphi$ an order-invariant $\MSO$-formula of vocabulary $\tau_\leq$. Let $\mathcal{C}$ be the class of $\tau$ structures that $\varphi$ defines.
We will show that $\mathcal{C}$ is definable in $\MSOTC$.
Set $n\dfn 2^k$ and let $Y_1,\dots, Y_{n}$ denote the Boolean combinations of the variables in $\tau$ in some fixed order; we regard these combinations also as fresh monadic second-order variables and set $\sigma \dfn \{Y_1,\dots,Y_n\}$. For each $X_i$, let $\chi_i$ denote the disjunction of those variables $Y_j$ in which $X_i$ occurs positively. Let $\mathcal{C}_\leq$ denote the class of $\tau_\leq$-structures that $\varphi$ defines. We may view $\mathcal{C}_\leq$ also as a language $\la$ over the alphabet $\sigma$ and as the class $\la_w$ of $\sigma_\leq$-structures corresponding to the word models of the language $\la$.
Let $\varphi^*$ denote the order-invariant $\MSO$-formula over $\sigma_\leq$ obtained from $\varphi$ by substituting each variable $X_i$ by the formula $\chi_i$. Since $\varphi^*$ clearly defines $\la_w$, by B\"uchi's Theorem, $\la$ is regular. Consequently, by the improved version of Parikh's Theorem (Theorem \ref{Parikhthm}), the Parikh image $\pa(\la)$ of $\la$ is a finite union of linear sets with at most $n$ generators.

%Consider the class $\mathcal{D}$ of $\sigma$-structures that $\varphi^*$ defines; remember that $\varphi^*$ is an order-invariant $\sigma_\leq$-formula.
%Observe that $\pa(\la)=\pa(\mathcal{D})$, and thus $\pa(\mathcal{D})$ is a finite union of linear sets with at most $n$ generators.
Observe that if two $\tau$-structures have the same Parikh image, the structures are isomorphic. Thus $\mathcal{C}$ is invariant under Parikh images. Hence $\mathcal{C}$ is uniquely characterised by its Parikh image $\pa(\mathcal{C})$, which, since $\pa(\la)=\pa(\mathcal{C})$, is  a finite union of linear sets with at most ${n}$ generators.

\textbf{Claim:} For every linear set $A\subseteq \N^n$, where $n=2^k$, there exists a formula $\varphi_A$ of $\MSOTC$ of vocabulary $\tau=\{X_1,\dots X_k\}$ such that $\varphi_A$ defines the class of $\tau$-structures that have $A$ as their Parikh image.

With the help of the above claim, the theorem follows in a straightforward manner.
Let $A_1,\dots,A_m$ be a finite collection of linear sets such that $\pa(\mathcal{C})=A_1\cup\dots\cup A_m$ and let $\varphi_{A_1},\dots, \varphi_{A_m}$ be the related $\MSOTC$-formulas of vocabulary $\tau$ provided by the claim. Clearly $\psi \dfn \varphi_{A_1}\lor\dots\lor \varphi_{A_m}$ defines $\mathcal{C}$.

%Now let $\psi$ denote $\MSOTC$-formula of vocabulary $\tau$ obtained from $\psi$ by treating the symbols $Y_j$ as Boolean combinations of the variables in $\tau$ (w.l.o.g., we may assume that variables in $\sigma$ do not occur as parameters to $\TC$-operators as those can be replace by quantifying fresh monadic variables). By construction $\psi$ defines $\mathcal{C}$.

\noindent\textbf{Proof of the Claim:} Let $A\subseteq \N^n$ be a linear set with $n$ generators, i.e.,
\[
A = \{\vec{v}_{0}  + \sum_{j=1}^n  a_j \vec{v}_{j} \mid a_1,\dots,a_n\in\N\}, \text{ for some $\vec{v}_{0},\vec{v}_{1},\dots,\vec{v}_{n}\in \N^n$}.
\]
For each tuple $\vec{v}\in\N^n$ and $n$-tuple of monadic second-order variables $\vec{Z}$, let $\size(\vec{Z},\vec{v})$ denote the $\FO$-formula stating that, for each $i$, the size of the extension of $\vec{Z}[i]$ is $\vec{v}[i]$.
For $0\leq i \leq n$, we introduce fresh distinct $n$-tuples of monadic variable symbols $\vec{Z}_i$ and define
\[
\varphi_{\mathrm{gen}} \dfn \bigwedge_{0 \leq i\leq n} \size(\vec{Z}_i,\vec{v}_i).
\]
Let $\vec{R}_1,\dots \vec{R}_n$ be fresh distinct n-tuples of monadic second-order variables and let $S_1,\dots, S_n$ be fresh distinct monadic second-order variables. Define
\begin{multline}
\varphi^*_A \dfn \exists \vec{Z}_0 \dots \vec{Z}_n \vec{R}_1 \dots \vec{R}_n S_1\dots S_n \, \varphi_{\mathrm{gen}} \land \\
 \bigwedge_{1\leq i,j \leq n} \times(\vec{Z}_i[j], S_i, \vec{R}_i[j]) \land  \bigwedge_{1\leq i \leq n}   +(\vec{Z}_0[i], \vec{R}_1[i], \dots, \vec{R}_n[i],Y_i),
\end{multline}
where  $\times$ and $+$ refer to the $\MSOTC$-formulas defined in Example \ref{ex:predicates}. 
Finally define $\varphi_A \dfn \exists Y_1\dots Y_n \, \varphi_{BC} \land \varphi^*_A$, where $\varphi_{BC}$ is an $\FO$-formula stating that, for each $i$, the extension of $Y_i$ is the extension of the Boolean combination of the variables in $\tau$ that $Y_i$ represents.
A $\tau$-structure $\mB$ satisfies $\varphi_A$ if and only if the Parikh image of $\mB$ is $A$.
\end{proof}

\section{Conclusion} \label{seconc}

There are quite a number of interesting challenging questions
regarding the expressive power within second-order transitive-closure
logic.
\begin{enumerate}
\item
We have shown that $\MSOTC$ can do counting, and thus can certainly
express some queries not expressible in fixpoint logic FO(LFP)\@.  
A natural question is whether $\MSOTC$ can also be separated
from the counting extension of FO(LFP)\@.  Note that $\MSOTC$ can
express numerical predicates in NLOGSPACE, while counting
fixpoint logic can express numerical predicates in PTIME\@.
Thus, over the empty vocabulary, the question seems related to a
famous open problem from complexity theory.  Note however, that
it is not even clear that $\MSOTC$ can \emph{only} express numerical
predicates in NLOGSPACE.
\item
The converse question, whether there is a fixpoint logic query
not expressible in $\MSOTC$, is fascinating.  On ordered
structures, this would show that there are problems in PTIME that are not
in NLIN, which is open (we only know that the two classes are
different \cite{papadimitriou}).  On unordered structures,
however, we actually
conjecture that the query about a binary relation (transition
system) $R$ and two nodes $a$ and $b$, that asks whether $a$ and
$b$ are \emph{bisimilar} w.r.t.\ $R$, is not expressible in $\MSOTC$.
\item
In stating Theorem~\ref{thm:fotc} we recalled that
$\SO(\mathrm{arity}\, k)(\TC)$ captures the complexity class
$\NSPACE(n^k)$, on strings. What about ordered structures in
general?  Using the standard adjacency matrix encoding of
a relational structure as a string \cite{immerman1998}, it
follows that on ordered structures over vocabularies with maximal arity $a$, 
$\SO(\mathrm{arity}\, {k\cdot a})(\TC)$ can express all queries
    in
$\NSPACE(n^k)$.  Can we show that this blowup in arity is
necessary?  For example, can we show that $\MSOTC$ does \emph{not}
capture NLIN over ordered graphs (binary relations)?
\item
In the previous section
we have clarified the relationship between 
$\MSOTC$ and order-invariant $\MSO$, over unary vocabularies.
What about higher arities?
\end{enumerate}

%% Bibliography
\bibliography{bibfile,bibfileFlavio}

\end{document}